\documentclass[
submission
]{dmtcs-episciences}

\usepackage[utf8]{inputenc}
\usepackage{subfigure}

\usepackage{complexity}
\usepackage{amssymb}
\usepackage{amsmath}
\usepackage[ruled,vlined,linesnumbered]{algorithm2e}
\usepackage{tikz}

\newcommand{\borda}[1]{\overset{\rightarrow}{#1}}
\newcommand{\semborda}[1]{\overset{\leftarrow}{#1}}

\newtheorem{theorem}{Theorem}[section]
\newtheorem{corollary}{Corollary}[section]
\newtheorem{lemma}{Lemma}[section]

\usepackage[round]{natbib}

\author[M.C. Dourado \and V.S. Ponciano \and R.L.O. da Silva]{Mitre C. Dourado\affiliationmark{1}\thanks{Partially supported by Conselho Nacional de Desenvolvimento Cient\'ifico e Tecnol\'ogico, Brazil, Grant number 305404/2020-2.}
  \and Vitor S. Ponciano\affiliationmark{1}\thanks{Partially supported by Coordena\c{c}\~ao de Aperfei\c{c}oamento de Pessoal de N\'ivel Superior, Brazil.}
  \and R\^omulo L. O. da Silva\affiliationmark{1,2}}
\title{On the monophonic rank of a graph}
\affiliation{
  Instituto de Computa\c{c}\~ao, Universidade Federal do Rio de Janeiro, Rio de Janeiro, Brazil\\
  Faculdade de Ci\^encia e Tecnologia, Instituto de Matem\'atica, Universidade Federal do Par\'a, Par\'a, Brazil}
\received{10$^{th}$ Oct. 2020}
\revised{16$^{th}$ Aug. 2022}
\accepted{17$^{th}$ Aug. 2022}
\begin{document}
\publicationdetails{24}{2022}{2}{3}{6835}	
\maketitle
\begin{abstract}
A set of vertices $S$ of a graph $G$ is \emph{monophonically convex} if every induced path joining two vertices of $S$ is contained in $S$. The \emph{monophonic convex hull of $S$}, $\langle S \rangle$, is the smallest monophonically convex set containing $S$. A set $S$ is \emph{monophonic convexly independent} if $v \not\in \langle S - \{v\} \rangle$ for every $v \in S$. The \emph{monophonic rank} of $G$ is the size of the largest monophonic convexly independent set of $G$. We present a characterization of the monophonic convexly independent sets. Using this result, we show how to determine the monophonic rank of graph classes like bipartite, cactus, triangle-free and line graphs in polynomial time. Furthermore, we show that this parameter can be computed in polynomial time for $1$-starlike graphs, \textit{i.e.}, for split graphs, and that its determination is $\NP$-complete for $k$-starlike graphs for any fixed $k \ge 2$, a subclass of chordal graphs. We also consider this problem on the graphs whose intersection graph of the maximal prime subgraphs is a tree.
\end{abstract}

\section{Introduction}

A family $\mathcal {C}$ of subsets of a finite set $X$ is a \emph{convexity on $X$} if $\varnothing, X \in\mathcal{C}$ and $\mathcal{C}$ is closed under intersection~(\cite{van1993theory}).
Given a graph $G$ and a family ${\cal P}$ of paths of $G$, the \emph{${\cal P}$-interval of a set $S \subseteq V(G)$} is formed by $S$ and all vertices of every path of ${\cal P}$ between vertices of $S$. The set $S$ is \emph{${\cal P}$-convex} if $S$ is equal to its ${\cal P}$-interval. The \emph{${\cal P}$-convex hull of $S$} is the minimum ${\cal P}$-convex set containing $S$. It is easy to see that the ${\cal P}$-convex sets form a convexity on $V(G)$. Indeed, the most studied graph convexities are defined in this way. For instance, the well-known \emph{geodetic convexity} has ${\cal P}$ as the family of shortest paths~(\cite{PELAYO,dou1}), the \emph{$P_3$ convexity} has ${\cal P}$ as the family of paths of order 3~(\cite{CAMPOS,CENTENO20113693}), and in the \emph{monophonic convexity}, ${\cal P}$ is the family of all induced paths~(\cite{DOURADO1,DUCHET1}).

The set $S$ is said to be \textit{${\cal P}$-convexly independent} if for every $u \in S$, it holds that $u$ does not belong to the ${\cal P}$-convex hull of $S - \{u\}$. The size of a maximum ${\cal P}$-convexly independent set of $G$ is called the \textit{${\cal P}$-rank of $G$}. 
We are interested in the computational complexity of determining the rank of a graph in the monophonic convexity.

\bigskip
\noindent {\sc Monophonic rank}\\
\begin{tabular}{lp{14.5cm}}
	Instance: & A graph $G$ and a positive integer $k$. \\
	Question: & Does $G$ contain a monophonic convexly
	independent set  with at least $k$ vertices?
\end{tabular}
\bigskip

In~\cite{ramos:hal-01185615}, it was shown that {\sc Monophonic rank} is \NP-complete for prime graphs, which are the graphs not containing clique separators. For trees, the maximum convexly independent set corresponds to the set of leaves.
In the same work, it was shown that the problem of deciding whether a graph has rank at least $k$ in the $P_3$ convexity is \NP-complete even for split graphs and for bipartite graphs.
On the other hand, it was shown that this parameter can be easily determined for threshold graphs in the $P_3$ convexity, and for trees in the $P_3$ and geodetic convexities.
In the $P_{3}$ convexity, this problem is solvable in polynomial time for biconnected graphs and it was conjectured that it can also be solved in polynomial time for interval graphs.

In~\cite{kante2017geodetic}, the authors studied the computational complexity of finding a maximum convexly independent set in the geodetic convexity. They showed that this problem cannot be approximated within a factor of $n^{1-\epsilon}$, unless $\P=\NP$. They also presented polynomial-time algorithms for computing this parameter for $P_4$-sparse, split graphs, and distance-hereditary graphs.

The text is organized as follows. 
Useful notation is presented in the end of this section.
In Section~\ref{sec:charac}, we present a characterization of the monophonic rank of a graph, which is useful in some results in the sequel.
In Section~\ref{sec:bipartite}, we show that the problem of computing the monophonic rank can be solvable in polynomial time if the graph is $K_4$-free and has clique separators with at most 2 vertices, which includes bipartite, cactus, and triangle-fre graphs.
In Section~\ref{sec:inter}, we deal with the graphs whose intersection graph of the maximal prime subgraphs is a tree, which includes line graphs.
In Section~\ref{sec:starlike}, we show that this problem is solvable in polynomial time for 1-starlike graphs, \textit{i.e.}, for split graphs, and is \NP-complete for $k$-starlike graphs for any fixed $k \ge 2$, a subclass of chordal graphs.

We conclude this section by presenting some useful notation.
We consider only simple, finite and undirected graphs.
The open and the closed neighborhoods of a vertex $v$ are denoted by $N(v)$ and $N[v]$, respectively. Vertices $u$ and $v$ are \emph{twins} if $N[u] = N[v]$.
The interval and the convex hull of a set $S$ in the monophonic convexity are denoted by $[S]$ and $\langle S \rangle$, respectively.
We can use m-convexly independent standing for monophonically convexly independent. The monophonic rank of a graph $G$ is denoted by $r(G)$. The size of a maximum clique of $G$ is denoted by $\omega(G)$.
Given a family of graphs ${\cal C}$, we denote the intersection graph of the members ${\cal C}$ by $\Omega({\cal C})$.

\section{Characterization} \label{sec:charac}

In this section, we present a characterization of the m-convexly independent sets of a general graph. This result is used in the algorithm of Section~\ref{sec:bipartite} to determine the monophonic rank of graph classes like bipartite, cactus and triangle-free graphs in polynomial time.

A connected graph is \emph{prime} if it does not contain any clique separator.
A \emph{maximal prime subgraph of $G$}, or simply an \emph{mp-subgraph of $G$}, is a maximal induced subgraph of $G$ that is prime. Given an mp-subgraph $M$, denote the vertices of $M$ belonging to other mp-subgraphs by $\borda{M}$, and denote $\semborda{M} = V(M) - \borda{M}$.
We say that $M$ is a \emph{petal} if $\semborda{M} \ne \varnothing$ and $\borda{M}$ is a clique.
We say that $M$ is \emph{extreme} if there is an mp-subgraph $M' \ne M$ such that $\borda{M} \subseteq \borda{M'}$.
Note that every extreme mp-subgraph is also a petal mp-subgraph.
We say that $H$ is an \emph{flower subgraph of $G$} if $H = G$ or $H = G' - \semborda{M}$ where $G'$ is a flower subgraph of $G$ and $M$ is a petal mp-subgraph of $G'$.
Denote by ${\cal F}(G), {\cal M}(G)$, and ${\cal P}(G)$ the families of the flower subgraphs of $G$, of the mp-subgraphs of $G$, and of the petal mp-subgraphs of $G$, respectively.

For $M \in {\cal P}(G)$, we say that $X \subseteq \semborda{M}$ is a \emph{stamen set for $M$ in $G$} if

\begin{itemize}
	\item $|X| = 1$ (Type~1); or
	\item $X = \{u_1,u_2\}$ and $\borda{M} \subset N(u_i)$ for $i \in \{1,2\}$ (Type~2); or
	\item $X \cup \borda{M}$ is a clique and $|X| \ge 2$ (Type~3).
\end{itemize}

\begin{theorem} \label{the:mhp} \emph{(\cite{DOURADO1})}
	Every two non-adjacent vertices of a prime graph form a monophonic hull set.
\end{theorem}

\begin{theorem} \label{the:charac}
	A set $S \subseteq V(G)$ of a connected graph $G$ is m-convexly independent if and only if for every mp-subgraph $M$ of the minimum flower subgraph $H$ of $G$ containing $\langle S \rangle$, if $M$ is a leaf, then $S \cap V(M)$ is contained in a stamen set of $M$ in $H$, otherwise $S \cap V(M) = \varnothing$.
\end{theorem}

\begin{proof}
	First, consider that $S$ is an m-convexly independent set of $G$.
	
	We begin showing that no vertex of $S$ belongs to $\borda{M}$ for any $M \in {\cal M}(H)$. Suppose the contrary and let $w \in \borda{M_1}$ for some mp-subgraph $M_1$ of $H$. Let $C_w$ be a clique separator of $H$ containing $w$ contained in $M_1$. Let $M_2 \ne M_1$ be an mp-subgraph containing $C_w$.
	Define $H_1$ as the maximum subgraph of $H$ containing $M_1$ and not containing $M_2 - C_w$. Analogously, define $H_2$ as the maximum subgraph of $H$ containing $M_2$ and not containing $M_1 - C_w$.
	By the minimality of $H$, there are $u_1 \in (S \cap V(H_1)) - C_w$ and $u_2 \in (S \cap V(H_2)) - C_w$.
	Now, let $P_1$ be an induced $(u_1,w)$-path of $H_1 - (C_w - \{w\})$ and let $P_2$ be an induced $(u_2,w)$-path of $H_2 - (C_w - \{w\})$.
	These two paths imply that $w \in \langle u_1,u_2 \rangle$, which contradicts the assumption that $S$ is an m-convexly independent of $G$.
	
	Suppose that there is an mp-subgraph $M \in {\cal M}(H) - {\cal P}(H)$ such that $w \in V(M) \cap S$.
	Let $u,v \in \borda{M}$ such that $uv \not\in E(G)$. Let $C_u$ and $C_v$ be clique separators of $H$ such that $u \in C_u \subset V(M)$ and $v \in C_v \subset V(M)$. Now, define $H_u$ as a maximum connected subgraph of $H$ containing $M$ as a leaf with $\borda{M} = C_u$, and define $H_v$ as a maximum connected subgraph of $H$ containing $M$ as a leaf with $\borda{M} = C_v$.
	
	Note that $H - V(M)$ is a disconnected graph. Therefore, by the minimality of $H$, there are $u' \in (S \cap V(H_u)) - V(M)$ and $v' \in (S \cap V(H_v)) - V(M)$.
	Now, let $P_u$ be an induced $(u',u)$-path of $H_u - (C_u - \{u\})$, let $P_v$ be an induced $(v',v)$-path of $H_v - (C_v - \{v\})$, and let $P$ be an induced $(u,v)$-path of $M$.
	These three paths imply that $u,v \in \langle u',v' \rangle$. Now, Theorem~\ref{the:mhp} implies that $w \in \langle u',v' \rangle$, contradicting the assumption that $S$ is an m-convexly independent of $G$.
	
	Now, suppose that there is $M \in {\cal P}(H)$ such that $W = V(M) \cap S$ is not a stamen set for $M$ in $H$. Suppose that $x,y \in W - \borda{M}$. We can assume that $yu \not\in E(G)$ for some $u \in \borda{M}$, which means that $x \in \langle u,y \rangle$ and since $u \in \langle u',y \rangle$ for some vertex $u' \in S - V(M)$, we have a contradiction.
	
	Conversely, let $w \in S$. By the definition of stamen set, we conclude that $w \in \semborda{M}$ of some petal mp-subgraph of $H$, \textit{i.e.}, $\borda{M}$ is a clique. If $V(M) \cap S = \{w\}$, it is clear that $w \not\in \langle S - \{w\} \rangle$, then we can assume that $|V(M) \cap S| \ge 2$ and $w \in \langle S - \{w\} \rangle$. But this contradicts the definition of stamen set, because in any case that $|V(M) \cap S| \ge 2$, it holds that $\borda{M} \subseteq N(w)$, which contradicts $w \in \langle S - \{w\} \rangle$.
\end{proof}

As a consequence, we can express the monophonic rank of a graph in terms of its flower subgraphs and the stamen sets of its petal mp-subgraphs.

\begin{corollary} \label{cor:mrank}
	For any graph $G$, $r(G) = \underset{G' \in {\cal F}(G)}{\max} \left\{\underset{M \in {\cal P}(G')}{\sum} s(M) \right\}$ where $s(M)$ stands for the maximum size of a stamen set of $M$ in $G'$.
\end{corollary}

Corollary~\ref{cor:mrank} is used in Sections~\ref{sec:bipartite} and~\ref{sec:inter} to solve {\sc Monophonic rank} in polynomial time for some graph classes.

\section{Bipartite, cactus and triangle-free graphs} \label{sec:bipartite}

In this section, we show how to compute the monophonic rank of a class containing the bipartite, cactus and triangle-free graphs in polynomial time. We say that a graph $G$ is \emph{bipartite} if $V(G)$ can be partitioned into two independent sets, and that $G$ is a \emph{cactus} if every maximal 2-connected subgraph of $G$ is a cycle or an edge.
We define $\Gamma_1$ as the class containing the graphs $G$ such that $G$ is $K_4$-free and if an mp-subgraph $B$ of $G$ contains a $K_3$, then $B$ is isomorphic to $K_3$.
Note that $\Gamma_1$ contains the bipartite, cactus and triangle-free graphs.
The following property of general graphs is important for this purpose.

\begin{lemma} \label{lem:numberofleaves}
	Let $G$ be any graph. For every flower subgraph $G'$ of $G$, it holds $|{\cal P}(G)| \ge |{\cal P}(G')|$.
\end{lemma}

\begin{proof}
	We use induction on the number of mp-subgraphs of $G$. First, consider that $G$ has 2 mp-subgraphs. Note that $|{\cal P}(G)| = 2$ and that $|{\cal P}(G')| = 1$ for any of its flower subgraphs. Now, consider that $G$ has $k$ mp-subgraphs for $k \ge 3$.
	Suppose that there is a flower subgraph $G'$ of $G$ with $|{\cal P}(G)| < |{\cal P}(G')|$. Since $G'$ can be obtained from $G$ by iteratively removing a petal mp-subgraph, there is a petal mp-subgraph $M$ of $G$ not present in $G'$. Define $G'' = G - \semborda{M}$.
	
	We claim that $|{\cal P}(G)| \ge |{\cal P}(G'')|$. If $M$ is not an extreme mp-subgraph, then there is exactly one mp-subgraph $M'$ in $G''$ that is not in $G$, namely, the one formed by the vertices of $\borda{M}$. Observe that ${\cal M}(G'') = ({\cal M}(G) - \{M\}) \cup \{M'\}$. Since the intersections among the mp-subgraphs of $G$ and $G''$ are the same, it holds that $G''$ has no new petal mp-subgraph and $|{\cal P}(G)| > |{\cal P}(G'')|$ for this case because $M'$ is not a petal mp-subgraph of $G''$.
	
	Now, if $M$ is an extreme mp-subgraph, then every mp-subgraph of $G''$ is also present in $G'$ because $\borda{M}$ is contained in an mp-subgraph $M''$. Then, ${\cal M}(G'') = {\cal M}(G) - \{M\}$. It is possible that $M''$ be a petal mp-subgraph of $G''$ not present in $G$, but this cannot happen to any other mp-subgraph because for each one its intersection is contained in mp-subgraphs also present in $G$. Since $M$ is a leaf of $G$ not present in $G''$, it holds $|{\cal P}(G)| \ge |{\cal P}(G'')|$, and the claim is true.
	
	Finally, by the induction hypothesis, we have $|{\cal P}(G'')| \ge |{\cal P}(G')|$, which is a contradiction.
\end{proof}

\begin{theorem} \label{the:Gamma}
	For a graph $G \in \Gamma_1$, $r(G) = n_1 + 2n_2$ where $n_1$ is the number of petal mp-subgraphs of Type~$1$, $n_2$ is the number of petal mp-subgraphs of Type~$2$ of the flower subgraph $G'$ of $G$ such that $|{\cal P}(G')| = |{\cal P}(G)|$ with maximum number of leaves of Type~$2$.
\end{theorem}

\begin{proof}
	First, observe that every petal mp-subgraph of a graph $G \in \Gamma_1$ has Type~1 or~2.
	Now, let $S$ be an m-convexly independent set of $G$ with $|S| = r(G)$ and let $H$ be the minimum flower subgraph containing the vertices of $\langle S \rangle$.
	By Corollary~\ref{cor:mrank}, for every petal mp-subgraph $M$ of $H$, $|S \cap \semborda{M}|$ is equal to the size of a maximum stamen set of $M$ in $H$.
	Since $\Gamma_1$ is a hereditary class, it holds $1 \le |S \cap V(M)| \le 2$.
	Hence, Lemma~\ref{lem:numberofleaves} implies that $H$ is the flower subgraph of $G$ with $|{\cal P}(H)| = |{\cal P}(G)|$ which maximizes the number of leaves of Type~$2$.
\end{proof}

We have the following complexity result as a consequence of the above theorem.

\begin{corollary}
	One can determine the monophonic rank of a graph $G \in \Gamma_1$ in polynomial time.
\end{corollary}

\begin{proof}
	Set $G'$ equals $G$. Now, while there is petal mp-subgraph $B$ of $G'$ having Type~1 such that $G'' = G' - \semborda{B}$ satisfies $|{\cal P}(G'')| = |{\cal P}(G')|$, set $G' = G' - \semborda{B}$ and repeat. The instance of $G'$ after the loop is the flower subgraph of $G$ with the same number of petal mp-subgraphs as $G$ with maximum number of leaves of Type~$2$. By Theorem~\ref{the:Gamma}, $r(G) = n_1 + 2n_2$ where $n_1$ is the number of petal mp-subgraphs of Type~$1$, $n_2$ is the number of petal mp-subgraphs of Type~$2$. It is clear that $G'$ can be obtained in polynomial time on the size of $G$.
\end{proof}

\section{When the intersection graph of the mp-subgraphs is a tree} \label{sec:inter}

In this section, we consider the class $\Gamma_2$ of the graphs $G$ such that $\Omega({\cal M}(G))$ is a tree. We will see that $\Gamma_2$ contains the line graphs. We begin characterizing the graphs of $\Gamma_2$. Then, we present in Section~\ref{sec:metaAlg} a meta-algorithm for trees. Such algorithm is part of the algorithm presented in the sequel for the computation of the monophonic rank of graphs of $\Gamma_2$. We decided to present the exploitation of the tree-like structure of our solution separatly because it is essentially the same that appears implicitly in other solutions for other problems~(\cite{ANAND2020,BENEVIDES201588}), which facilitates its use in future works.

\begin{lemma} \label{lem:intGraph}
	A graph $G \in \Gamma_2$ if and only if every vertex of $G$ belongs to at most $2$ mp-subgraphs of $G$.
\end{lemma}

\begin{proof}
	Denote $H = \Omega({\cal M}(G))$. First, consider that $v$ is a vertex belonging to 3 mp-subgraphs $M_1, M_2,$ and $M_3$ of $G$. It is clear that the vertices $u_1, u_2,$ and $u_3$ of $H$ corresponding to $M_1, M_2,$ and $M_3$, respectively, induce a $C_3$, which means that $H$ is not a tree.
	
	Conversely, consider that $H$ has a cycle $C = u_1 \ldots u_k$. We can assume that $C$ is induced. If $k \ge 4$, then $G - M_i \cap M_{i+1}$ is a connected graph for $i+1$ taken mod $k$. Therefore, $M_1 \cup \ldots \cup M_k$ has no clique separator and properly contains an mp-subgraph. Therefore, $k = 3$ and suppose that every vertex of $G$ belongs to at most $2$ mp-subgraphs of $G$. Hence. $(M_i \cap M_j) \cap (M_i \cap M_k) = \varnothing$ for $i,j,k$ being different values of $\{1,2,3\}$. Therefore, $G - M_i \cap M_{i+1}$ is a connected graph for $i+1$ taken mod $3$. Therefore, $M_1 \cup \ldots \cup M_3$ has no clique separator and properly contains an mp-subgraph, which is a contradiction.	
\end{proof}

\subsection{Meta-algorithm for trees} \label{sec:metaAlg}

Let $T$ be a tree of order $n$ and let $F$ and $F'$ be sets of functions. We say that $(T,F,F')$ is a \emph{good triple} if

\begin{itemize}
	\item $|F| = n$;
	
	\item $|F'| = 2n-2$;
	
	\item for every $v \in V(T)$ and $w \in N(v)$, there is $f'_{v,u} \in F'$ which depends of the values of $f'_{w,v}$ for $w \in N(v) - \{u\}$ and does not depend of any other function of $F \cup F'$; and
	
	\item for every $v \in V(T)$, there is $f_{v} \in F$ which depends of the values of $f'_{w,v}$ for $w \in N(v)$ and does not depend of any other function of $F \cup F'$.
\end{itemize}

If $v$ is a leaf of $T$, we denote the only neighbor of $v$ in $T$ by $\pi_T(v)$.

\begin{algorithm}[h] \label{alg:general}
	
	\caption{{\sc Meta\_Algorithm\_Tree}}
	
	\KwIn{A good triple $(T,F,F')$}
	\KwOut{Computation of $f$ and $f'$ for every $v \in V(T)$}
	
	\If{$V(T) = \{w\}$}{
		compute $f_w$
		
		\Return \label{lin:n1}
	}
	
	$T' \leftarrow T$
	
	\While{$|V(T')| \ge 3$}{ \label{lin:1WhileBegin}
		$w \leftarrow$ leaf of $T'$
		
		compute $f'_{w,\pi_{T'}(w)}$ \label{lin:f'A}
		
		remove $w$ of $T'$ \label{lin:1WhileEnd}
	}
	
	$w \leftarrow$ one of the two vertices of $T'$
	
	compute $f'_{w,\pi(w)}$ \label{lin:n2}
	
	compute $f'_{\pi(w),w}$
	
	\While{there is leaf $w$ of $T'$ with $f_w$ not computed yet}{ \label{lin:2WhileBegin}
		
		\For{$u \in N_T(w)$ such that $u \ne \pi_{T'}(w)$}{ \label{lin:ForBegin}
			compute $f'_{w,u}$ \label{lin:f'B}
			
			add $u$ to $T'$ as neighbor of $w$ \label{lin:ForEnd}
		}
		compute $f_w$ \label{lin:f}
	}
\end{algorithm}

\begin{theorem} \label{the:metaalg}
	If $(T,F,F')$ is a good triple, Algorithm~$\ref{alg:general}$ computes all functions of $F$ and $F'$ in time $O( n (\alpha + \alpha'))$ steps, where $\alpha$ and $\alpha'$ are the time complexities of computing $f \in F$ and $f' \in F'$, respectively, and $n$ is the number of vertices of the input tree.
\end{theorem}

\begin{proof}
	The number of iterations of the {\bfseries while} beginning in line~\ref{lin:1WhileBegin} is $n - 2$. The cost of each iteration is $\alpha'$. Therefore, the total cost of lines~\ref{lin:1WhileBegin} to~\ref{lin:1WhileEnd} is $O(n \alpha')$.
	It remains to show that $f'_{w,\pi_{T'}(w)}$ can be computed in line~\ref{lin:f'A}. But this is consequence of the fact that $f'_{w,\pi_{T'}(w)}$ is computed only if $w$ is a leaf of $T'$ at this moment, \textit{i.e.}, $f'_{u,w}$ has already been computed for every $u \in N_T(w) - \{\pi_{T'}(w)\}$.
	
	The number of iterations of the {\bfseries while} loop beginning in line~\ref{lin:2WhileBegin} is $n$. The cost of each line~\ref{lin:f} is $\alpha$.
	Even inside of the {\bfseries while} loop beginning in line~\ref{lin:2WhileBegin}, the total number of iterations of the {\bfseries for} loop beginning in line~\ref{lin:ForBegin} is $n - 2$. The cost of each iteration is $\alpha'$.
	Therefore, the total cost of lines~\ref{lin:2WhileBegin} to~\ref{lin:f} is $O(n (\alpha + \alpha'))$.
	It remains to show that $f'_{w,\pi_{T'}(w)}$ can be computed in line~\ref{lin:f'B} and that $f_w$ can be computed in line~\ref{lin:f}. Both cases are consequences of the fact that $f'_{u,w}$ has already been computed for every $u \in N_T(w)$ when $w$ is chosen in line~\ref{lin:2WhileBegin}.
\end{proof}

\subsection{Applying the meta-algorithm}

Now, we define the functions of a good triple of a graph $G \in \Gamma_2$ for expressing the monophonic rank of $G$.
Write $T = \Omega({\cal M}(G))$. For every $w \in V(T)$ and $u \in N(w)$, denote by $M_w$ the mp-subgraph of $G$ corresponding to $w$ and by $M_{w,u}$ the petal mp-subgraph of the graph associated to the tree of $T - (N(w) - u)$ containing $w$. Define

\begin{equation} \label{equ:1int}
f'_{w,u} = \max \left\{ s(M_{w,u}), \underset{v \in N(w) - \{u\}}{\sum} f'_{v,w} \right\}
\end{equation}

\begin{equation} \label{equ:2int}
f_w = \max \left\{ \omega(M_w) , \underset{v \in N(w)}{\sum} f'_{v,w} \right\}
\end{equation}

\begin{theorem} \label{the:inter}
	If $G \in \Gamma_2$, then $r(G) = \underset{w \in V(\Omega({\cal M}(G)))}{\max} \{f_w\}$.
\end{theorem}

\begin{proof}
	Write $T = \Omega({\cal M}(G))$. By Corollary~\ref{cor:mrank}, $r(G) = \underset{G' \in {\cal F}(G)}{\max} \left\{\underset{M \in {\cal P}(G')}{\sum} s(M) \right\}$ where $s(M)$ stands for the maximum size of a stamen set of $M$ in $G'$. Let $G_r \in {\cal F}(G)$ such that $r(G) = \underset{M \in {\cal P}(G_r)}{\sum} s(M)$. Note that $f'_{M,\pi_{T_r}(M)} = s(M)$ for every $M \in {\cal P}(G_r)$. For every $w$, since $f_w$ is defined as the maximum between two values where one of them is the sum of $f'_{u,w}$ for all $u \in N(w)$, for every $w$ that is not a leaf of $T_r$, it holds that $f_w = \underset{M \in {\cal P}(G_r)}{\sum} s(M)$. If $G_r$ is not a prime graph, then $T_r$ has at least one non-leaf vertex. Otherwise, $f_w$ is the maximum clique of the graph prime $G_r$, which corresponds to the maximum degree of a vertex of $G$.
\end{proof}

The following result is a consequence of Theorems~\ref{the:metaalg} and~\ref{the:inter}.

\begin{corollary} \label{cor:gamma2}
	If $G \in \Gamma_2$, then $r(G)$ can be computed in $O(n (m + \alpha + \alpha'))$ steps where $\alpha$ and $\alpha'$ are the time complexities of determining the clique number of $G$ and a maximum stamen set of a petal mp-subgraph of a flower subgraph of $G$, respectively.
\end{corollary}

\begin{proof}
	Denote $T = \Omega({\cal M}(G))$. It is easy to see that $(T, F, F')$ is a good triple.
	Using Theorem~\ref{the:inter}, it remains to prove the time complexity.
	By Theorem~\ref{the:metaalg}, $f'_{w,u}$ and $f_w$ can be computed in $O(n (\alpha + \alpha'))$ steps.
	Since ${\cal M}(G)$ can be computed in $O(nm)$ steps~(\cite{Leimer1993}), the result does follow.
\end{proof}

We say that a graph class $\Gamma$ is \emph{hereditary} if $G \in \Gamma$ implies that every induced subgraph of $G$ also belongs to $\Gamma$. As a consequence of Corollary~\ref{cor:gamma2}, the monophonic rank can be computed in polynomial time for every graph $G \in \Gamma_2 \cap \Gamma$ such that $\Gamma$ is hereditary and the clique number can be computed in polynomial time for the graphs of $\Gamma$, for instance, when $\Gamma$ is the class of chordal graphs. We will see in Section~\ref{sec:starlike} that {\sc Monophonic rank} is \NP-complete for chordal graphs.

We conclude this section by showing that there are non-hereditary graph classes $\Gamma$ so that {\sc Monophonic rank} is \NP-complete even if the clique number can be computed in polynomial time.

We define $\Gamma_3$ as the class containing the graphs $G$ such that $|V(G)| = 2n+1$ for some integer $n \ge 1$, $V(G)$ can be partitioned into sets $(V_1,V_2,V_3)$ where $|V_1| = n$, $V_1$ induces a clique, $V_2$ induces a connected graph with no universal vertices, $V_3 = \{u\}$ such that $N(u) = V_2$, and for every $v \in V_1$, there is $w$ such that $N(v) \cap V_2 = \{w\}$ and $N(w) \cap V_1 = \{v\}$.

It is easy to see that every graph of $\Gamma_3$ is prime, can be recognized in polynomial time, and its maximum clique can also be found in polynomial time.

\begin{theorem}
	{\sc Monophonic rank} is \NP-complete for the graphs whose mp-subgraphs belong to $\Gamma_3$.
\end{theorem}

\begin{proof}
	We present a reduction from a version of {\sc Maxclique} that receives as input a connected graph $G$ or order $n$, with no universal vertices, and asks whether there exists a clique in $G$ with at least $\lceil\frac{n}{2} \rceil$ vertices. Let $H$ be the graph with vertex set $(V_1,V_2,V_3,V_4,\{w\})$, such that $H[V_2] \simeq H[V_4] \simeq G$, $V_1$ and $V_3$ are cliques, $N(w) = V_2 \cup V_4$, for every $v \in V_1$, there is $w$ such that $N(v) \cap V_2 = \{w\}$ and $N(w) \cap V_1 = \{v\}$, and the same for $V_3$ and $V_4$. Notice that $H$ contains two mp-subgraphs, each one belonging to $\Gamma_3$. It is easy to see that $rk(H) = \max \{n, \omega(H[V_2]) + \omega(H[V_4])\}$.
\end{proof}

\subsection{Line graphs} \label{sec:line}

Let $G$ be a graph and $v \in V(G)$. The line graph of $G$, denoted by $L(G)$, is $\Omega(E(G))$. We say that $G$ is a \emph{line graph} if there is a graph $H$ such that $G = L(H)$. In this case, $L^{-1}(G)$ represents the family of graphs that are a root of $G$ under the line graph operator.
Denote by $E(v)$ the set of edges of $G$ incident to $v$.
If $v$ is a cut vertex of $G$ and $H$ is a connected component of $G - v$, then the subgraph of $G$ induced by $V(H) \cup \{v\}$ is a \emph{$v$-component}.

\begin{lemma} \label{lem:2Cprime}
	If $G$ is a $2$-connected graph, then $L(G)$ is prime.
\end{lemma}

\begin{proof}
	Suppose by contradiction that $K$ is a clique separating vertices $u,v \in L(G)$.
	We can write $u = u_1u_2$ and $v = v_1v_2$ for $u_1,u_2,v_1,v_2 \in V(G)$.
	Since $K$ separates $u$ and $v$, we have that $u_1,u_2,v_1$ and $v_2$ are distinct.
	Since $G$ is 2-connected, there is a cycle $C$ containing both edges $u_1u_2$ and $v_1v_2$.
	Without loss of generality, we can assume that $u_1$ and $v_1$ separate $u_2$ and $v_2$ in the subgraph of $G$ induced by the edges of $C$.
	Then, 
	let $P_{u_1u_2v_1v_2}$ be the $(u_1,v_2)$-path formed by the edges of $C$ containing the edges $u_1u_2,v_1v_2$, and let $P_{u_2u_1v_2v_1}$ be the $(u_2,v_1)$-path formed by the edges of $C$ containing the edges $u_1u_2,v_1v_2$.
	Now, denote by $E'$ the edges of $G$ associated with the vertices of $K$.
	There are two possibilities for the elements of $E'$,
	either they are the edges of a triangle of $G - \{u_1u_2,v_1v_2\}$
	or they are some edges of $E(G) - \{u_1u_2,v_1v_2\}$ incident to a vertex $w$.
	In both cases, at least one of the paths $P_{u_1u_2v_1v_2}$ and $P_{u_2u_1v_2v_1}$ does not contain edges of $E'$, which is a contradiction.
\end{proof}

Define ${\cal L}(G) = \{L(B) : B$ is a block of $G\} \cup \{L(E(v)) : v$ is a cut vertex of $G\}$.

\begin{lemma} \label{lem:mp-subgraphs}
	Let $G$ be a graph. Then, the family of mp-subgraphs of $L(G)$ is ${\cal L}(G)$. Furthermore, if $|{\cal L}(G)| \ge 2$ and $B$ is a block of $G$ with at least $3$ vertices, then $|L(\borda{B})| \ge 2$.
\end{lemma}

\begin{proof}
	By Lemma~\ref{lem:2Cprime}, $L(B)$ is a prime graph for every block $B$ of $G$.	Since $L(E(v))$ is a clique, then $L(E(v))$ also induces a prime subgraph for every cut vertex $v$.
	The maximality of each $M \in {\cal L}$ comes from the fact that
	$M \cup S$ for $S \subset H$ where $H$ is a $v$-component for $v \in M$ has $E(v) \cap H$ as a clique separator.
	Suppose by contradiction that $L(G)$ has an mp-subgraph $M' \not\in {\cal L}$. Therefore, $M'$ contains vertices that are edges of different blocks of $G$ sharing a cut vertex $v$ and at least one of these edges, say $uw$, is not incident to $v$. Now observe that $E(v) \cap H$, where $H$ is the $v$ component containing $uw$ is a clique separator of $M'$, which is a contradiction.
	
	Finally, let $|{\cal L}(G)| \ge 2$ and $B$ be a block of $G$ with at least $3$ vertices. Since there is a cut vertex $v$ of $G$ belonging to $B$, it holds $|L(\borda{B})| \ge 2$ because $v$ has at least 2 neighbors in $B$.
\end{proof}

\begin{lemma} \label{lem:inter}
	If $G$ is a line graph, then $G \in \Gamma_2$.	
\end{lemma}

\begin{proof}
	By Lemma~\ref{lem:mp-subgraphs}, ${\cal M}(G) = {\cal L}(G)$.
	Now, the result follows from the fact that every vertex of $G$ belongs to at most two sets of ${\cal L}(G)$ and Lemma~\ref{lem:intGraph}.
\end{proof}

\begin{corollary}
	The monophonic rank of a line graph can be found in linear time.
\end{corollary}

\begin{proof}
	We begin showing that for line graphs, every stamen set of a petal mp-subgraph $M$ of a line graph $G$ that is not prime is $1$. Then, suppose the contrary and let $u,v \in \semborda{M}$ and $w \in \borda{M}$ such that $uw, vw \in E(G)$. We can write $u = u_1u_2$, $v = v_1v_2$ and $w = w_1w_2$ where $u_1,u_2,v_1,v_2,w_1$ and $w_2$ are vertices of $G^L \in L^{-1}(G)$. Without loss of generality, we can assume that $u_1 = v_1 = w_1$ while $u_2,v_2$ and $w_2$ are distinct. Therefore, $w_2$ is a cut vertex of $G^L$ and $u$ has no more neighbors in $\borda{M}$, which implies that $|\borda{M}| = 1$. But this a contradiction because Lemma~\ref{lem:mp-subgraphs} implies that $|\borda{M}| \ge 2$.
	
	If $G$ is a prime graph, then $G^L$ is a 2-connected graph and the maximum stamen set of $G$ is the maximum degree of $G^L$.
	
	Finally, Corollary~\ref{cor:mrank} and Lemma~\ref{lem:numberofleaves} imply that we only need to consider one flower subgraph of $G$, namely, itself. Therefore, $r(G) = \max \left\{ \Delta(G^L), \ell ) \right\}$, where $\ell$ stands for the number of leaf blocks of $G^L$. Since a member of $G^L$ can be found in linear time~(\cite{lehot1974optimal}), the result does hold.
\end{proof}

\section{Starlike graphs} \label{sec:starlike}

A graph $G$ is \emph{starlike} if $V(G)$ can be partitioned into cliques $(V_0, \ldots, V_t)$ such that $V_0$ is a maximal clique and for every $\ell \in \{1, \ldots, t\}$ and $v_i,v_j \in V_{\ell}$, it holds $N[v_i] - V_\ell = N[v_j] - V_\ell \subset V_0$. If $|V_\ell| \le k$ for $\ell \in \{1, \ldots, t\}$, then $G$ is \emph{$k$-starlike}~(\cite{Gustedt:1993,Cerioli2006}). The 1-starlike graphs are precisely the split graphs and every starlike graph is a chordal graph.

In this section, we describe a polynomial-time algorithm for determining the monophonic rank of a $1$-starlike graph and show that this problem is $\NP$-complete for $k$-starlike graphs for any fixed $k \ge 2$.

\bigskip
\noindent {\sc Independent set} \\
\begin{tabular}{lp{14.5cm}}
	Instance: & A graph $G$ and a positive integer $k$. \\
	Question: & Does $G$ contain an independent set with at least $k$ vertices?
\end{tabular}
\bigskip

\begin{theorem}\label{sec:1starlike}
	The {\sc Monophonic rank} problem restricted to split graphs belongs to $\P$. 
\end{theorem}

\begin{proof}
	Since {\sc Independent set} restricted for bipartite graphs belongs to $\P$~(\cite{Garey1}), it suffices to present a polynomial reduction from {\sc Monophonic rank} restricted to split graphs to {\sc Independent set} restricted to bipartite graphs.
	
	Let $G$ be a split graph with bipartition $(C,I)$. We can assume that $C$ is a maximum clique. Construct a bipartite graph $G'$ from a copy of $G$ by deleting the edges joining vertices of $C$. Denote $(C',I')$ the bipartition of $G'$, where $C' = \{v'_i : v_i \in C\}$. See Figure~\ref{fig:splitP}. We will show that $G$ has an m-convexly independent set of size at least $k$ if and only if $G'$ has an independent set of size at least $k$.
	
	\begin{figure}[h]
		\centering
		
		\begin{tikzpicture}[scale=0.6]
		
		\pgfsetlinewidth{1pt}
		
		\tikzset{
			vertex/.style={circle,  draw, minimum size=5pt, inner sep=0pt}}
		
		\begin{scope}[shift={(25,0)}]
		
		\draw (7,5) node[above] {$G'$};
		
		\node [vertex] (v1) at (5,4) [label=left:$v_1$]{};
		\node [vertex, fill=black] (v2) at (5,2) [label=left:$v_2$]{}; 
		\node[vertex, fill=black]  (v3) at (5,0) [label=left:$v_3$]{};
		\node[vertex, fill=black] (v4) at (5,1) [label=left:$v_4$]{};
		\node [vertex] (v5) at (5,-2) [label=left:$v_5$]{};
		\node [vertex, fill=black]  (v6) at (5,-3) [label=left:$v_5$]{}  ;
		\node[vertex, fill=black] (v7) at (5,-4) [label=left:$v_5$]{};
		
		\node [vertex] (w1) at (10,2) [label=right:]{} edge (v1);
		\node [vertex, fill=black] (w1') at (10,4) [label=right:]{} edge (v1);
		\node [vertex, fill=black] (w1'') at (10,4.5) [label=right:]{} edge (v1);
		\node [vertex, fill=black] (w1'') at (10,5) [label=right:]{} edge (v1);
		\node [vertex] (w2) at (10,2) [label=right:]{} edge (v2) edge (v4) edge (v3);
		
		\node [vertex, fill=black] (w4) at (10,-4) [label=right:]{}  edge(v5);
		
		\node [vertex] (w3) at (10,-0.5) [label=right:]{} edge (v3);
		\node [vertex, fill=black] (w4) at (10,-3.5) [label=right:]{}  edge(v5);
		
		\draw[dashed] (5,0) ellipse (2cm and 5.5cm);
		\draw[dashed] (10,0) ellipse (2cm and 5.5cm);
		
		\end{scope}
		
		\begin{scope}[shift={(20,0)}]
		\draw (-2,5) node[above] {$G$};
		\node [vertex] (v'1) at (-4,4) [label=left:$v_1$]{};
		\node [vertex, fill=black]  (v'2) at (-4,2) [label=left:$v_2$]{} ;
		\node [vertex, fill=black]  (v'3) at (-4,1) [label=left:$v_3$]{};
		\node [vertex, fill=black] (v'4) at (-4,0) [label=left:$v_4$]{};
		\node [vertex] (v'5) at (-4,-2) [label=left:$v_5$]{};
		\node [vertex, fill=black]  (v'6) at (-4,-3) [label=left:$v_5$]{} ;
		\node [vertex, fill=black]  (v'7) at (-4,-4) [label=left:$v_5$]{} ;
		
		\node [vertex, fill=black] (w'1') at (1,4) [label=right:]{} edge (v'1);
		\node [vertex, fill=black] (w'1'') at (1,4.5) [label=right:]{} edge (v'1);
		\node [vertex, fill=black] (w'1'') at (1,5) [label=right:]{} edge (v'1);
		\node [vertex] (w'1) at (1,2) [label=right:]{} edge (v'1);
		\node [vertex] (w'2) at (1,2) [label=right:]{} edge (v'2) edge (v'4) edge (v'3)  ;
		
		\node [vertex, fill=black] (w'4) at (1,-4) [label=right:]{}  edge(v'5);
		
		\node [vertex] (w'3) at (1,-0.5) [label=right:]{} edge (v'4);
		\node [vertex, fill=black] (w'4) at (1,-3.5) [label=right:]{} edge(v'5);
		
		\draw[-] (-4,0) ellipse (2cm and 5.7cm);
		\draw[dashed] (1,0) ellipse (2cm and 5.7cm);
		\end{scope}
		\end{tikzpicture}
		
		\caption{Reduction from {\sc Monophonic rank problem} restricted to split graphs to {\sc Independent set problem} restricted to bipartite graphs. The vertices inside of the oval with a straight form a clique.}
		\label{fig:splitP}
		
	\end{figure}
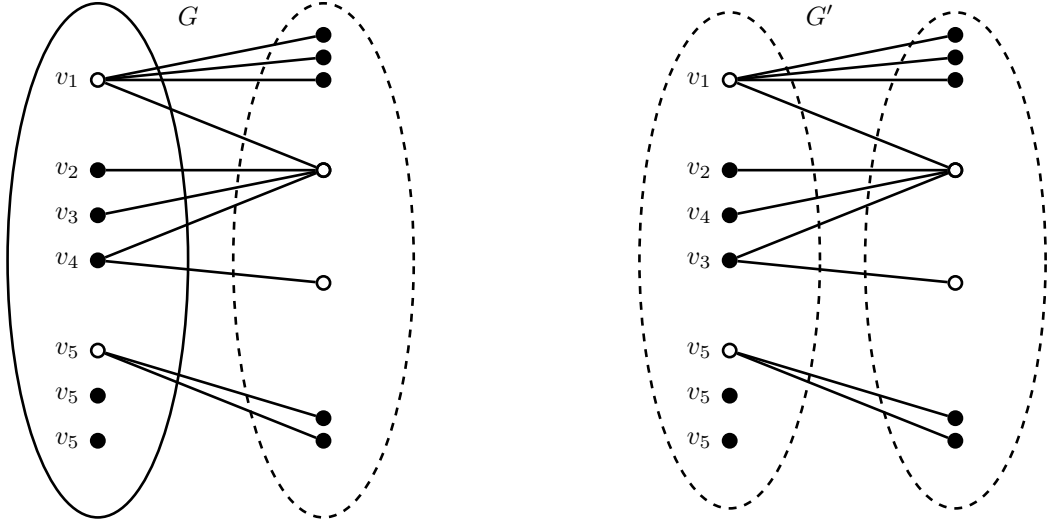	
	
	First, let $S$ be an m-convexly independent set of size at least $k$ of $G$. If $|S| \leq |C|$, then $G'$ has an independent set of size at least $k$ contained in $C'$ because $C'$ is an independent set. Then, consider $|S| > |C|$ and define $S' = \{v'_i : v_i \in S\}$. Suppose that $v'_iv'_j \in E(G)$ for $v'_i,v'_j \in S'$. By the construction of $G'$, we can assume that $v_i \in C$ and $v_j \in I$. Since $S$ is not a clique, there is $v_\ell$ such that $v_\ell v_j \not\in E(G)$. Now, observe that $v_i$ belongs to a monophonic path between $v_j$ and $v_\ell$, which contradicts the assumption that $S$ is an m-convexly independent set of $G$.
	
	Conversely, let $S'$ be an independent set of $G'$ with at least $k$ vertices. Write $S=\{v_i:v'_{i}\in S'\}$. Observe that
	
\begin{displaymath}\langle S - \{v_i\} \rangle = (S - \{v_i\}) \bigcup \left( \underset{v_j\in(S - \{v_i\})\cap I}{\bigcup}N(v_j) \right).
\end{displaymath}

\noindent Since $S'$ is an independent set, we have that $v_i \not\in \langle S - \{v_i\} \rangle$, which means that $S$ is an m-convexly independent set of $G$.

\end{proof}

Now, we show that the problem of deciding whether a $2$-starlike graph has a convexly independent set with at least $k$ vertices is \NP-complete in the monophonic convexity. It is clear that the following variation of {\sc Independent set} is also $\NP$-complete~\cite{Garey1}.

\bigskip
\noindent {\sc Half independent set} \\
\begin{tabular}{lp{14.5cm}}
	Instance: & A graph $G$. \\
	Question: & Does $G$ contain an independent set with at least $\left\lceil\frac{|V(G)|+1}{2}\right\rceil$ vertices?
\end{tabular}
\bigskip

\begin{theorem}\label{the:NPC}
	{\sc Monophonic rank} is \NP-complete for $2$-starlike graphs.
\end{theorem} 

\begin{proof}
	Since the monophonic convex hull of a set can computed in polynomial time and one needs $|S|$ computations of the convex hull of a set to decide whether a set is convexly independent, the problem belongs to \NP. We present a reduction from {\sc Half independent set}.
	
	Let $G$ be a general graph with $n = |V(G)|$. We construct a $2$-starlike graph $G'$ as follows. For every vertex $v_i \in V(G)$, create $4n+1$ vertices $u_i^1, \ldots u_i^{2n+1},w_i^1, \ldots, w_i^{2n}$ in $G'$. Write $U = \{u_i^1, \ldots, u_i^{2n+1}\} : 1 \leq i \leq n\}$ and $W = \{w_i^1, \ldots, w_i^{2n} : 1 \leq i \leq n\}$. The edge set of $G'$ is obtained as follows 
	
	\begin{itemize}
		\item add the edges to make a clique of the set $U$;
		
		\item add the edges $u_i^{2n+1}w_i^{2n-1}$ and $u_i^{2n+1}w_i^{2n}$ for every $v_i \in V(G)$; and
		
		\item for every $v_iv_j \in E(G)$, add the edges $u_i^p w_j^q$ for $p,q \in \{1, \ldots, 2n\}$.	
	\end{itemize}
	
	\begin{figure}[h]
		\centering
		
		\begin{tikzpicture}[scale=0.55]
		
		\pgfsetlinewidth{1pt}
		
		\tikzset{vertex/.style={circle,  draw, minimum size=5pt, inner sep=1pt}}
		
		\draw (-9.5,5) node[above] {$G$};
		
		\node [vertex,fill=black] (v1) at (-9,4) [label=left:$v_1$]{};
		\node [vertex] (v2) at (-9,3) [label=left:$v_2$]{} edge (v1) edge (v2);
		\node [vertex,fill=black]  (v3) at (-9,2) [label=left:$v_3$]{} edge (v2);
		
		\draw (-4.5,9) node[above] {$G'$};
		
		
		
		\def\H{-1}
		\def\V{0}
		\def\i{1}
		
		\node [vertex, fill=black] (u11) at (\H + -1,\V + 7) [label=above:$u^1_\i$]{};
		\node [vertex, fill=black] (u12) at (\H + 0,\V + 7) [label=above:$u^2_\i$]{};
		\node [vertex, fill=black]  (u13) at (\H + 1,\V + 7) [label=above:$u^3_\i$]{};
		\node[vertex, fill=black]  (u14) at (\H + -1,\V + 6) [label=below:$u^4_\i$]{};
		\node [vertex, fill=black] (u15) at (\H + -0,\V + 6) [label=below:$u^5_\i$]{} ;
		\node[vertex, fill=black]  (u16) at (\H + 1,\V + 6) [label=below:$u^6_\i$]{};
		\node[vertex]  (u17a) at (\H + 3,\V + 6.5) [label=below:$u^7_\i$]{};
		\draw[dashed][black] (\H + 0,\V + 6.5) ellipse (2.5cm and 2.5cm);	
		
		\def\H{-1}
		\def\V{-6}
		\def\i{2}
		
		\node [vertex] (u11) at (\H + -1,\V + 7) [label=above:$u^1_\i$]{};
		\node [vertex] (u12) at (\H + 0,\V + 7) [label=above:$u^2_\i$]{};
		\node [vertex]  (u13) at (\H + 1,\V + 7) [label=above:$u^3_\i$]{};
		\node[vertex]  (u14) at (\H + -1,\V + 6) [label=below:$u^4_\i$]{};
		\node [vertex] (u15) at (\H + -0,\V + 6) [label=below:$u^5_\i$]{} ;
		\node[vertex]  (u16) at (\H + 1,\V + 6) [label=below:$u^6_\i$]{};
		\node[vertex, fill=black]  (u17b) at (\H + 3,\V+ 6.5) [label=below:$u^7_\i$]{};
		\draw[dashed][black] (\H + 0,\V + 6.5) ellipse (2.5cm and 2.5cm);	
		
		\def\H{-1}
		\def\V{-12}
		\def\i{3}
		
		\node [vertex, fill=black] (u11) at (\H + -1,\V + 7) [label=above:$u^1_\i$]{};
		\node [vertex, fill=black] (u12) at (\H + 0,\V + 7) [label=above:$u^2_\i$]{};
		\node [vertex, fill=black]  (u13) at (\H + 1,\V + 7) [label=above:$u^3_\i$]{};
		\node[vertex, fill=black]  (u14) at (\H + -1,\V + 6) [label=below:$u^4_\i$]{};
		\node [vertex, fill=black] (u15) at (\H + -0,\V + 6) [label=below:$u^5_\i$]{} ;
		\node[vertex, fill=black]  (u16) at (\H + 1,\V + 6) [label=below:$u^6_\i$]{};
		\node[vertex]  (u17) at (\H + 3,\V + 6.5) [label=below:$u^7_\i$]{};
		\draw[dashed][black] (\H + 0,\V + 6.5) ellipse (2.5cm and 2.5cm);	
		
		\draw[][black] (-1,0.5) ellipse (5cm and 10cm);

		
		
		\def\H{9}
		\def\V{0}
		\def\i{1}
		
		\node [vertex, fill=black] (u11) at (\H + -1,\V + 7) [label=above:$w^5_\i$]{}edge(u17a);
		\node [vertex, fill=black] (u12) at (\H + 0,\V + 7) [label=above:$w^3_\i$]{};
		\node [vertex, fill=black]  (u13) at (\H + 1,\V + 7) [label=above:$w^1_\i$]{};
		\node[vertex, fill=black]  (u14) at (\H + -1,\V + 6) [label=below:$w^6_\i$]{}edge(u11) edge(u17a);
		\node [vertex, fill=black] (u15) at (\H + -0,\V + 6) [label=below:$w^4_\i$]{} edge(u12);
		\node[vertex, fill=black]  (u16) at (\H + 1,\V + 6) [label=below:$w^2_\i$]{}edge(u13);
		\draw[dashed][black] (\H + 0,\V + 6.5) ellipse (2.5cm and 2.5cm);	
		
		\def\H{9}
		\def\V{-6}
		\def\i{2}
		
		\node [vertex] (u11) at (\H + -1,\V + 7) [label=above:$w^5_\i$]{}edge(u17b);
		\node [vertex] (u12) at (\H + 0,\V + 7) [label=above:$w^3_\i$]{};
		\node [vertex]  (u13) at (\H + 1,\V + 7) [label=above:$w^1_\i$]{};
		\node[vertex]  (u14) at (\H + -1,\V + 6) [label=below:$w^6_\i$]{}edge(u11) edge(u17b);
		\node [vertex] (u15) at (\H + -0,\V + 6) [label=below:$w^4_\i$]{} edge(u12);
		\node[vertex]  (u16) at (\H + 1,\V + 6) [label=below:$w^2_\i$]{}edge(u13);
		\draw[dashed][black] (\H + 0,\V + 6.5) ellipse (2.5cm and 2.5cm);	
		
		\def\H{9}
		\def\V{-12}
		\def\i{3}
		
		\node [vertex, fill=black] (u11) at (\H + -1,\V + 7) [label=above:$w^5_\i$]{}edge(u17);
		\node [vertex, fill=black] (u12) at (\H + 0,\V + 7) [label=above:$w^3_\i$]{};
		\node [vertex, fill=black]  (u13) at (\H + 1,\V + 7) [label=above:$w^1_\i$]{};
		\node[vertex, fill=black]  (u14) at (\H + -1,\V + 6) [label=below:$w^6_\i$]{}edge(u11) edge(u17);
		\node [vertex, fill=black] (u15) at (\H + -0,\V + 6) [label=below:$w^4_\i$]{} edge(u12);
		\node[vertex, fill=black]  (u16) at (\H + 1,\V + 6) [label=below:$w^2_\i$]{}edge(u13);
		\draw[dashed][black] (\H + 0,\V + 6.5) ellipse (2.5cm and 2.5cm);	
		
		\def\V{0}
		\draw[] (1,\V + 5) to (7.7,\V + 2.6);
		\draw[] (1,\V + 2) to (6.7,\V + 5.6);
		\def\V{-6}
		\draw[] ( 1,\V + 5) to (7.7, \V + 2.6);
		\draw[] ( 1,\V +  2) to (6.7,\V + 5.6);
		
		\end{tikzpicture}
		
		\caption{Sketch of the reduction from {\sc Half independent set problem} of general graph to {\sc Monophonic rank}  of 2-starlike graphs. A line joining two cycles with dashed lines means that there is an edge joining each vertex of one cycle to each vertex of the other.}
		\label{fig:my_label}
		
	\end{figure}
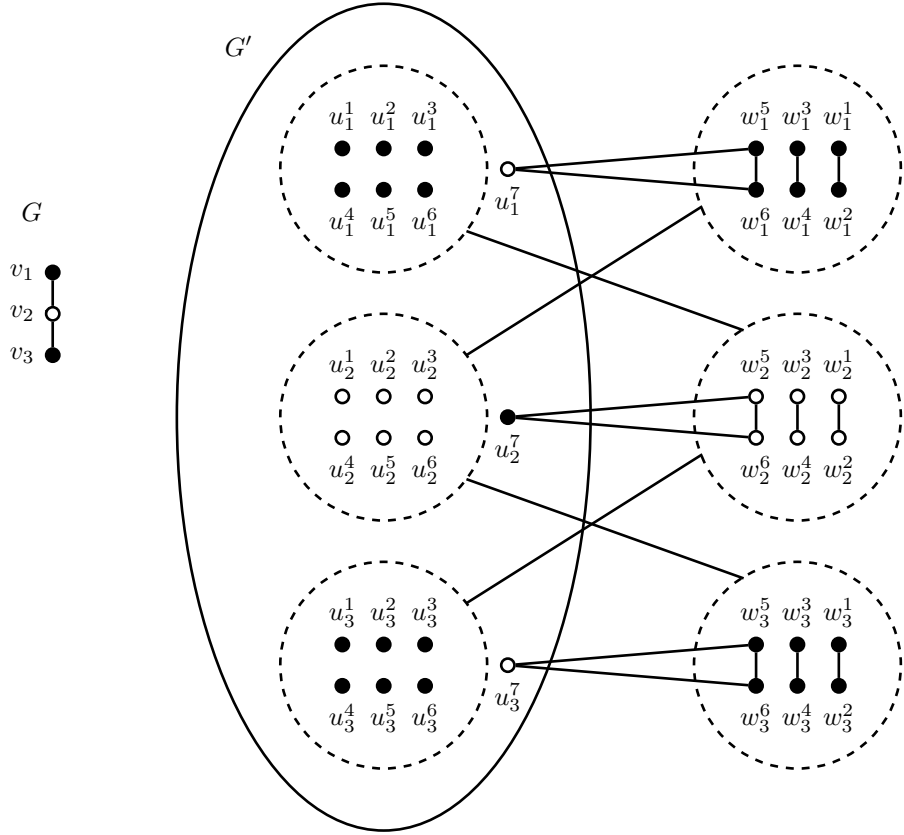
	
Since $U$ is a clique and for $i,\ell \in \{1, \ldots, n\}$, the vertices $w_i^{2\ell}$ and $w_i^{2\ell+1}$ are twins and each one has only one neighbor in $W$, it holds that $G'$ is a $2$-starlike graph. Note that the number of vertices of $G'$ is $4n^2+n$, which means that $G'$ can be constructed in polynomial time. To complete the instance of {\sc Monophonic rank} define

\begin{displaymath}
p = 4n \left\lceil\frac{n+1}{2}\right\rceil + n - \left\lceil\frac{n+1}{2}\right\rceil = n + (4n-1)\left\lceil\frac{n+1}{2}\right\rceil.
\end{displaymath}	
		
We will show that $G$ has an independent set with $\lceil\frac{n+1}{2}\rceil$ vertices if and only if $G'$ has a monophonic convexly independent set with $p$ vertices. Since every vertex has a neighbor in the set $U$, it is easy to see that the diameter of $G'$ is 3 (and a maximum induced path has length~3).
	
First, consider that $G$ has an independent set $S$ with $\left\lceil\frac{n+1}{2}\right\rceil$ vertices. Define $S' = \{u_i^1, \ldots, u_i^{2n}$, $w_i^1, \ldots, w_i^{2n} : v_i \in S\} \cup \{u_i^{2n+1} : v_i \not\in S\}$. Since $S'$ has $4n$ vertices for every vertex in $S$ and $1$ for each vertex not in $S$, it holds

\begin{displaymath}
|S'| = 4n \left\lceil\frac{n+1}{2}\right\rceil + (n - \left\lceil\frac{n+1}{2}\right\rceil) = p.
\end{displaymath}
 
Suppose that there is $u \in S'$ such that $u \in \langle S' - \{u\} \rangle$. Hence, $u$ belongs to some monophonic $(x,x')$-path for $x,x' \in \langle S' - \{u\} \rangle$. Since $W$ contains only simplicial vertices and no simplicial vertex is an internal vertex of an induced path, we conclude that $u \in U \cap S'$. Since $U$ is a clique, at least one of $x$ and $x'$ belongs to $W$. However, observe that for every $S'' \subseteq S'$, it holds $\langle S'' \rangle \subseteq (S'' \cap W) \cup U$, in other words, no vertex of $W - S''$ belongs to $\langle S'' \rangle$. Then, we can assume that $x \in W \cap S'$. By the construction of $S'$, no vertex of $W \cap S'$ has neighbors in $U$, which implies that $ux \not\in E(G')$. Therefore, $x'$ also belongs to $W \cap S'$. Since the maximum induced path of $G'$ has 4 vertices, $ux' \in E(G)$, which is a contradiction.
	
Conversely, assume that $G'$ has an m-convexly independent set $S'$ with at least $p$ vertices. Since $p > |U| > |W|$ for $n \ge 1$, $S'$ contains vertices of $U$ and of $W$. For $j \in \{0, \ldots, 4n+1\}$, denote by $Q_j \subseteq V(G)$ the set of vertices $v_i$ such that $S'$ contains $j$ vertices of the $4n+1$ vertices of $G'$ created for $v_i$. Observe that $S = \underset{i \in \{3, \ldots, 4n\}}{\bigcup} Q_i$ is an independent set and that $Q_{4n+1} = \varnothing$ because otherwise $S'$ would not be m-convexly independent. Suppose that $S < \lceil\frac{n+1}{2}\rceil$. Hence,

\begin{displaymath}
|S'| \le 4n\left(\left\lceil\frac{n+1}{2}\right\rceil-1 \right) + 2n = 4n\left\lceil\frac{n+1}{2}\right\rceil -2n < n + (4n-1)\left\lceil\frac{n+1}{2}\right\rceil = p,
\end{displaymath}

\noindent which is a contradiction.
\end{proof}

The above result can be extended for $k$-starlike with $k \ge 2$.

\begin{corollary}
	For any fixed integer $k \ge 2$, the {\sc monophonic rank} problem is \NP-complete for $k$-starlike graphs.
\end{corollary}

\begin{proof}
	It suffices to change the reduction presented in Theorem~\ref{the:NPC} by adding $kn$ vertices in $G'$ for every vertex of $G$ and defining $p = n + kn \left\lceil\frac{n+1}{2}\right\rceil$.
\end{proof}

\nocite{*}
\bibliographystyle{abbrvnat}
\bibliography{mrank}

\begin{thebibliography}{16}
\providecommand{\natexlab}[1]{#1}
\providecommand{\url}[1]{\texttt{#1}}
\expandafter\ifx\csname urlstyle\endcsname\relax
  \providecommand{\doi}[1]{doi: #1}\else
  \providecommand{\doi}{doi: \begingroup \urlstyle{rm}\Url}\fi

\bibitem[Anand et~al.(2020)Anand, Anil, Changat, Dourado, and Ramla]{ANAND2020}
B.~S. Anand, A.~Anil, M.~Changat, M.~C. Dourado, and S.~S. Ramla.
\newblock Computing the hull number in ${\Delta}$-convexity.
\newblock \emph{Theoretical Computer Science}, 844:\penalty0 217--226, 2020.

\bibitem[Benevides et~al.(2015)Benevides, Campos, Dourado, Sampaio, and
  Silva]{BENEVIDES201588}
F.~Benevides, V.~Campos, M.~C. Dourado, R.~M. Sampaio, and A.~Silva.
\newblock The maximum time of 2-neighbour bootstrap percolation: algorithmic
  aspects.
\newblock \emph{European Journal of Combinatorics}, 48:\penalty0 88--99, 2015.

\bibitem[Campos et~al.(2015)Campos, Sampaio, Silva, and Szwarcfiter]{CAMPOS}
V.~Campos, R.~M. Sampaio, A.~Silva, and J.~L. Szwarcfiter.
\newblock Graphs with few ${P}_4$’s under the convexity of paths of order
  three.
\newblock \emph{Discrete Applied Mathematics}, 192:\penalty0 28--39, 2015.

\bibitem[Centeno et~al.(2011)Centeno, Dourado, Penso, Rautenbach, and
  Szwarcfiter]{CENTENO20113693}
C.~C. Centeno, M.~C. Dourado, L.~D. Penso, D.~Rautenbach, and J.~L.
  Szwarcfiter.
\newblock Irreversible conversion of graphs.
\newblock \emph{Theoretical Computer Science}, 412\penalty0 (29):\penalty0
  3693--3700, 2011.

\bibitem[Cerioli and Szwarcfiter(2006)]{Cerioli2006}
M.~R. Cerioli and J.~L. Szwarcfiter.
\newblock Characterizing intersection graphs of substars of a star.
\newblock \emph{Ars Combinatoria}, 79:\penalty0 21--31, 2006.

\bibitem[Dourado et~al.(2010)Dourado, Protti, and Szwarcfiter]{DOURADO1}
M.~C. Dourado, F.~Protti, and J.~L. Szwarcfiter.
\newblock Complexity results related to monophonic convexity.
\newblock \emph{Discrete Applied Mathematics}, 158\penalty0 (12):\penalty0
  1268--1274, 2010.

\bibitem[Dourado et~al.(2016)Dourado, de~S{\'a}, Rautenbach, and
  Szwarcfiter]{dou1}
M.~C. Dourado, V.~G.~P. de~S{\'a}, D.~Rautenbach, and J.~L. Szwarcfiter.
\newblock Near-linear-time algorithm for the geodetic radon number of grids.
\newblock \emph{Discrete Applied Mathematics}, 210:\penalty0 277--283, 2016.

\bibitem[Duchet(1988)]{DUCHET1}
P.~Duchet.
\newblock Convex sets in graphs, ii. minimal path convexity.
\newblock \emph{Journal of Combinatorial Theory, Series B}, 44\penalty0
  (3):\penalty0 307--316, 1988.

\bibitem[Garey and Johnson(1982)]{Garey1}
M.~R. Garey and D.~S. Johnson.
\newblock Computers and intractability: a guide to the theory of
  np-completeness (michael r. garey and david s. johnson).
\newblock \emph{Siam Review}, 24\penalty0 (1):\penalty0 90, 1982.

\bibitem[Gustedt(1993)]{Gustedt:1993}
J.~Gustedt.
\newblock On the pathwidth of chordal graphs.
\newblock \emph{Discrete Applied Mathematics}, 45\penalty0 (3):\penalty0
  233--248, 1993.

\bibitem[Kant{\'e} et~al.(2017)Kant{\'e}, Sampaio, dos Santos, and
  Szwarcfiter]{kante2017geodetic}
M.~M. Kant{\'e}, R.~M. Sampaio, V.~F. dos Santos, and J.~L. Szwarcfiter.
\newblock On the geodetic rank of a graph.
\newblock \emph{Journal of Combinatorics}, 8\penalty0 (2):\penalty0 323--340,
  2017.

\bibitem[Lehot(1974)]{lehot1974optimal}
P.~G. Lehot.
\newblock An optimal algorithm to detect a line graph and output its root
  graph.
\newblock \emph{Journal of the ACM (JACM)}, 21\penalty0 (4):\penalty0 569--575,
  1974.

\bibitem[Leimer(1993)]{Leimer1993}
H.-G. Leimer.
\newblock Optimal decomposition by clique separators.
\newblock \emph{Discrete mathematics}, 113\penalty0 (1-3):\penalty0 99--123,
  1993.

\bibitem[Pelayo(2013)]{PELAYO}
I.~M. Pelayo.
\newblock \emph{Geodesic convexity in graphs}.
\newblock Springer, 2013.

\bibitem[Ramos et~al.(2014)Ramos, dos Santos, and
  Szwarcfiter]{ramos:hal-01185615}
I.~Ramos, V.~dos Santos, and J.~Szwarcfiter.
\newblock Complexity aspects of the computation of the rank of a graph.
\newblock \emph{Discrete Mathematics and Theoretical Computer Science}, 16, 01
  2014.

\bibitem[van~de Vel(1993)]{van1993theory}
M.~van~de Vel.
\newblock \emph{Theory of Convex Structures}.
\newblock North-Holland Mathematical Library. Elsevier Science, 1993.
\newblock ISBN 9780444815057.

\end{thebibliography}
\label{sec:biblio}

\end{document}